\documentclass[conference]{IEEEtran}
\usepackage{cite}
\usepackage{amsmath,amssymb,amsfonts,amsthm}
\usepackage{graphicx}
\usepackage{textcomp}
\usepackage{algpseudocode, algorithm}
\usepackage{tabularx,multirow}
\usepackage{flushend}

\usepackage{epstopdf}

\usepackage{upgreek}
\usepackage{array}
\newcolumntype{L}[1]{>{\raggedright\let\newline\\\arraybackslash\hspace{0pt}}m{#1}}

\newtheorem{theorem}{Theorem}
\newtheorem{lemma}[theorem]{Lemma}
\newtheorem{corollary}[theorem]{Corollary}
\newcommand{\f}[2]{ f^{#1}_{[#2]} }
\newcommand{\ff}{ F }

\begin{document}
	
	\title{Transaction Propagation on Permissionless Blockchains: Incentive and Routing Mechanisms}

%\author{\IEEEauthorblockN{O\u{g}uzhan Ersoy\IEEEauthorrefmark{1},
%		Zhijie Ren\IEEEauthorrefmark{2}, Zekeriya Erkin\IEEEauthorrefmark{3} and
%		Reginald L. Lagendijk\IEEEauthorrefmark{4}}
%	\IEEEauthorblockA{Cyber Security Group, Department of Intelligent Systems,\\
%		Delft University of Technology\\
%		Delft, The Netherlands\\
%		Email: \IEEEauthorrefmark{1}o.ersoy@tudelf.nl,
%		\IEEEauthorrefmark{2}z.ren@tudelf.nl,
%		\IEEEauthorrefmark{3}z.erkin@tudelf.nl,
%		\IEEEauthorrefmark{4}r.l.lagendijk@tudelf.nl}}
	
	\author{\IEEEauthorblockN{O\u{g}uzhan Ersoy,
			Zhijie Ren, Zekeriya Erkin and
			Reginald L. Lagendijk}
		\IEEEauthorblockA{Cyber Security Group, Department of Intelligent Systems,\\
			Delft University of Technology\\
			Delft, The Netherlands\\
			Email: o.ersoy@tudelft.nl,
			z.ren@tudelft.nl,
			z.erkin@tudelft.nl,
			r.l.lagendijk@tudelft.nl}}
	
%\author{\IEEEauthorblockN{O\u{g}uzhan Ersoy}
%	\and
%		\IEEEauthorblockN{Zhijie Ren}
%		\and
%		\IEEEauthorblockN{Zekeriya Erkin}
%		\and
%			\IEEEauthorblockN{Reginald L. Lagendijk} 
%			\and
%	\IEEEauthorblockA{\textit{Department of Intelligence Systems} 
%		\textit{Delft University of Technology},
%		Delft, The Netherlands \\
%		o.ersoy@tudelf.nl}
%	\and
%
%	\IEEEauthorblockA{\textit{Department of Intelligence Systems} \\
%		\textit{Delft University of Technology}\\
%		Delft, The Netherlands \\
%		z.ren@tudelf.nl}
%	\and
%
%	\IEEEauthorblockA{\textit{Department of Intelligence Systems} \\
%		\textit{Delft University of Technology}\\
%		Delft, The Netherlands \\
%		z.erkin@tudelf.nl}
%	\and
%
%	\IEEEauthorblockA{\textit{Department of Intelligence Systems} \\
%		\textit{Delft University of Technology}\\
%		Delft, The Netherlands \\
%		r.l.lagendijk@tudelf.nl}
%}

\maketitle

\begin{abstract}
Existing permissionless blockchain solutions rely on peer-to-peer propagation mechanisms, where nodes in a network transfer transaction they received to their neighbors. 
Unfortunately, there is no explicit incentive for such transaction propagation. 
Therefore, existing propagation mechanisms will not be sustainable in a fully decentralized blockchain with rational nodes. 
In this work, we formally define the problem of incentivizing nodes for transaction propagation. 
We propose an incentive mechanism where each node involved in the propagation of a transaction receives a share of the transaction fee. 
We also show that our proposal is Sybil-proof. 
Furthermore, we combine the incentive mechanism with smart routing to reduce the communication and storage costs at the same time. 
The proposed routing mechanism reduces the redundant transaction propagation from the size of the network to a factor of average shortest path length. 
The routing mechanism is built upon a specific type of consensus protocol where the round leader who creates the transaction block is known in advance. 
Note that our routing mechanism is a generic one and can be adopted independently from the incentive mechanism.
\end{abstract}


\section{Introduction}
% Blockchain technology provides distributed trust platform and promises to replace centralized systems.
% A blockchain can be defined as a decentralized immutable public ledger which is updated and secured in a distributed structure among the untrusted parties.
% Permissionless blockchains, where anyone can immediately join and contribute to the chain, compose a large part of the existing blockchain applications, e.g. cryptocurrencies.
% A permissionless blockchain is composed of peer-to-peer network structure of nodes.
% Nodes (miners) validate each piece of information and agree on the state of the ledger by mining and consensus mechanisms.
% %Identities of the nodes are mostly not known, i.e.,  they are pseudonymous or anonymous participants.
% Each node stores the public ledger and they broadcast each information via a propagation mechanism.
% %The characteristics of the information propagation mechanism have not been investigated in depth.

%First, we analyze the necessary and sufficient conditions of Sybil-proof and incentive compatible propagation mechanism.
%We formulate generic rewarding function which encourages rational participants to propagate each information.
%Second, regarding the bandwidth efficiency, we study a specific type of consensus protocols where the block owner (round leader) is validated before the block is created.
%We present a smart routing mechanism which reduces the redundant communication cost from the size of the network to the scale of average shortest path length.

In this work, we investigate transaction propagation on permissionless blockchains with respect to incentive compatibility and bandwidth efficiency. 
%The first concern that we tackle in this paper is the incentive compatibility of the propagation mechanisms.
%The transition from a centralized system into a distributed one brings along game theoretical concerns \cite{monderer1999distributed}. 
%Indeed, rational behavior of the participants has been observed in peer-to-peer networks \cite{huberman1997social,adar2000free,Shneidman2003}.
The former, incentive compatibility, is an essential component of permissionless blockchain to maintain its functionality with rational participants \cite{peters2016understanding, Sompolinsky:2018:BUI:3190347.3152481}. The latter, bandwidth efficiency, is an important factor for efficient use of limited resources available in the network.   

Although a number of works have studied incentive compatibility problem of blockchains, they are limited to mining mechanism, e.g. investigating  \emph{selfish mining attacks} \cite{eyal2014majority,sapirshtein2016optimal,nayak2016stubborn,DBLP:conf/ccs/GervaisKWGRC16}, and \emph{block withholding attacks} \cite{DBLP:journals/corr/abs-1112-4980,DBLP:journals/corr/CourtoisB14,DBLP:conf/sp/Eyal15,DBLP:journals/tifs/BagRS17}.
The existing blockchain solutions such as Bitcoin \cite{nakamoto2008bitcoin} and Ethereum \cite{wood2014ethereum} do not pay attention to incentives for transaction propagation in the network. This is due to the fact that the mining networks in those solutions are centralized in practice \cite{gervais2014bitcoin,miller2015discovering,gencer2018decentralization} and thus, they do not exhibit a fully decentralized structure.
%The lack of incentive compatible propagation seems to be underrated since the most successful permissionless blockchain applications, including Bitcoin \cite{nakamoto2008bitcoin} and Ethereum \cite{wood2014ethereum}, work without any direct incentive to propagate. 
%Conversely, existing blockchain solutions, including Bitcoin \cite{nakamoto2008bitcoin} and Ethereum \cite{wood2014ethereum},  rely on an altruistic propagation mechanism \cite{altruism} in a network of centralized mining pools \cite{gervais2014bitcoin,miller2015discovering,gencer2018decentralization}.
%Because of their centralized mining networks \cite{gervais2014bitcoin,miller2015discovering,gencer2018decentralization}, they do not exhibit fully decentralized blockchain structure and their propagation mechanisms do not conflict with the necessity of the incentive.
%In other words, the existing solutions do not exhibit fully decentralized blockchain properties and do not conflict with the necessity of the incentive.
There are only two works that address incentive compatibility of transaction propagation in blockchain by Babaioff et al.\ \cite{babaioff2012bitcoin} and Abraham et al.\ \cite{abraham2016solidus}. Unfortunately, both works suggest a specific solution for the incentive compatibility but do not provide a formal definition of the problem. Furthermore,  the proposed solutions are also designed for certain network topologies.

In terms of bandwidth inefficiency, existing solutions suffer from multiple broadcasting of the same transaction over the network. 
For example, in Bitcoin, each transaction is received by the nodes (miners) in the network twice: once during the advertisement, i.e.\ broadcasting of the transaction at the beginning, and once after the validation, i.e.\ broadcasting of the block including the transaction. While validation is essential since each node in the network stores every validated transaction, the advertisement does not need to be received by all nodes. However, redundancy for advertisement is inevitable in such cases where the round leader who creates the validated block is unknown in advance since the transaction needs to be broadcast to all potential round leaders. In recent blockchain proposals where the round leader is known in advance, what we call \textit{first-leader-then-block} (FLTB) type of consensus protocols \cite{eyal16bitcoin,bentov2016cryptocurrencies,bentov2014proof,ouroboros},  it is possible to improve bandwidth efficiency by reducing the communication cost by directly routing the transaction to the round leader. To the best of our knowledge, there is no prior work on optimizing bandwidth efficiency for fully decentralized blockchain.

%of the information dissemination.
%It is necessary to disseminate all validated information through the blockchain network since the public ledger is stored and validated by all contributing nodes. 
%However, in the existing blockchains, there is an additional broadcasting of the information before being validated.
%This broadcasting advertises the transaction to the miners.
%For some consensus protocols as in Bitcoin, the redundancy is inevitable because of the unpredictability of the node (round leader) who will create the new block. 
%This is caused by the fact that the block owner is simultaneously validated with his proposed block. 
%Alternatively, in cases where the round leader is known in advance, what we prefer to call \textit{first-leader-then-block} (FLTB) type of consensus protocols \cite{eyal16bitcoin,bentov2016cryptocurrencies,bentov2014proof,ouroboros},        

%\subsection{Our Contributions}

In this work, our contribution is three-fold: 1) Sybil-proof incentive compatible propagation mechanism, 2) bandwidth-efficient routing mechanism, and 3) bandwidth and storage efficient transaction propagation mechanism which combines the first two mechanisms.

We formally define incentive compatibility of propagation mechanisms in fully decentralized blockchain networks.
We show that there is no Sybil-proof and incentive compatible propagation mechanism for poorly connected networks (specifically for 1-connected networks).
For other network topologies, we find the following incentive compatible and Sybil-proof formula, which distributes the transaction fee among propagating nodes:
	\begin{IEEEeqnarray*}{ll}
	\f{k}{i}= \begin{cases}
		\ff\cdot C(1-C)^{i-1} &\text{ for } i<k, \\
		\ff\cdot (1-C)^{k-1} &\text{ for } i=k,
	\end{cases}
\end{IEEEeqnarray*}
where $\ff$ is the fee, $k$ is the length of the propagation path, $\f{k}{i}$ is the share of the $i^{th}$ node in that path, and $C$ is a parameter related to the network topology.
The incentive mechanism is independent of the choice of consensus protocol and works with any consensus protocol.

We propose a routing mechanism compatible with FLTB-type consensus protocols. Our proposal reduces the communication cost of the transaction propagation from the size of the network to the scale of average shortest path length. 
In a random network topology of more than 500 nodes, we achieve over $97\%$ communication cost reduction compared to de facto propagation mechanism for the advertisement.
Furthermore, we also present a propagation mechanism which combines our incentive and routing mechanisms in a storage and bandwidth efficient way.
For incentive mechanism, our combined protocol requires storing only a single signature to provide the integrity of the path, unlike the existing works, which use a signature chain including signatures of each node in the path. 

The rest of the paper is organized as follows: 
Section \ref{sec:relatedwork} presents the related work.
Our blockchain model and notations are defined in Section \ref{sec:notation}.
Section \ref{sec:incentive} formulates requirements of the incentive problem and computes the generic solution. Smart routing mechanism is presented in Section \ref{sec:routing} and combined with incentive mechanism in Section \ref{sec:combined}. 
%Finally, Section \ref{sec:conc} concludes the paper.

\section{Related Work}\label{sec:relatedwork}

The lack of incentive for information propagation in a peer-to-peer network has been known and studied in different settings \cite{drucker2012simpler,emek2011mechanisms,kleinberg2005query,li2009incentive}. 
Kleinberg and Raghavan \cite{kleinberg2005query} proposed an incentive scheme for finding the answer for a given query in a tree-structured network topology. 
Li et al.\ \cite{li2009incentive} focused on node discovery in a homogeneous network where each node has the same probability of having an answer for the query. 
In \cite{drucker2012simpler,emek2011mechanisms}, the authors analyzed the incentive problem for multi-level marketing which rewards referrals if the advertisement produces a purchase.
In these marketing models, the reward is shared among all nodes in the tree including the propagation path.

%Conversely, in \cite{abraham2016solidus,babaioff2012bitcoin} which are focused on the blockchain, it is shared between only the ones in the propagation path which is the direct path between the client and the leader.

The proposed solutions for peer-to-peer networks \cite{drucker2012simpler,emek2011mechanisms,kleinberg2005query,li2009incentive} are not applicable for the permissionless blockchains.
In peer-to-peer solutions, nodes are asked to provide a specific datum like the position of a peer or the answer to a query.
In blockchains, however, transaction propagation is requested to advertise the transactions and eventually place them into a valid block.
Alternatively, finding an answer to a query is equivalent to validation of a transaction by round leader in the blockchain.
Query propagation in a peer-to-peer network has two main differences compared to a blockchain transaction propagation: nodes do not compete against the ones who forwarded the message to them and nodes cannot generate a response to a query that they do not have the answer, i.e.\, either they have the right answer or not.
Whereas in a blockchain, a block is generated by the round leader and every node is a potential round leader. 
Essentially, nodes in a blockchain are competitors that have an incentive not to propagate whereas other peer-to-peer nodes do not have the incentive since they cannot generate the answer to the query by themselves.

%In peer-to-peer solutions, participants are asked to provide a specific datum like the position of a peer or the answer to a query.
%In blockchains, it is requested to validate the transactions and place them into a valid block.
%In the latter case, participants can always compete for the reward without propagating, whereas in the former one, peers who have the datum do not need to propagate anymore and the others must propagate to have a chance.
%Essentially, nodes in a blockchain are competitors that have incentive not to propagate whereas other peer-to-peer nodes do not have the incentive since they cannot generate the required information by themselves.
Recently, blockchain oriented propagation mechanisms have been proposed \cite{abraham2016solidus,babaioff2012bitcoin}.
%Babaioff et al. \cite{babaioff2012bitcoin} presented the lack of incentive to propagate transactions in the Bitcoin, and they provided a solution for the $d$-ary directed tree networks. 
%In \cite{abraham2016solidus}, Abraham et al. investigated the propagation of not only transactions but also blocks.
In \cite{babaioff2012bitcoin}, Babaioff et al. uncovered the incentive problem in the Bitcoin system where a rational node (miner) has no incentive to propagate a transaction. 
They focused on a specific type of network, namely regular $d$-ary directed tree with a height $H$, and assumed that each node has the same processing power.
In that setting, the authors proposed a hybrid incentive (rewarding) scheme and proved that it is also Sybil-proof.
In \cite{abraham2016solidus}, Abraham et al. proposed a consensus mechanism, Solidus, offering an incentive to propagate transactions and validated blocks (puzzles). 
In their incentive mechanism, the amount of processing fee passed to the next node is determined by the sender.
%, and they concluded to be a fixed amount at the end of their analysis.
%It requires detailed study on determining the charging amount since small value would not lead to propagation and high value may not have enough incentive for block owner. 
Both works rely on a signature chain to prevent any manipulation over the path and thereby, to secure the shares of propagating nodes.

\cite{babaioff2012bitcoin} and \cite{abraham2016solidus} provided analyses of their proposals based on game theory.
For the analysis, \cite{babaioff2012bitcoin} assumes a tree-structured network which eliminates the case of competition against common neighbors and it is not realistic for blockchain network topology.
Whereas, the analysis in \cite{abraham2016solidus} is limited to the case of competition between nodes that have common neighbors.

%Regarding the game theoretical analysis, former one \cite{babaioff2012bitcoin} assumes tree structure which eliminates competition for the common neighbors, whereas the latter one \cite{abraham2016solidus} analyzes only the case of competitions between nodes for shared neighbors.
%In this paper, we analyze both cases in a generic network model without restricting the incentive mechanism. 

For bandwidth efficiency, to the best of our knowledge, there is no prior work for fully decentralized blockchain without dedicated miners (round leaders).
Nevertheless, Li et al. \cite{li2009incentive} presented a distributed routing scheme for peer-to-peer networks.
The authors focused on one-to-one routing which is dedicated to a single target node.
Whereas in blockchain it needs to be one-to-all routing, which connects the complete network to the round leader.
In addition, \cite{li2009incentive} does not take into account the possibility of a failing routing caused by a failing or malicious node in the routing path.
%In addition, unlike \cite{li2009incentive}, we insert alternative routes as a precaution and analyze the failure probability of these temporary routes caused by momently failing nodes.
%Existing static networks solutions cannot be applied since the round leader who is responsible to gather the transactions is frequently changing.

\section{Our Blockchain Model and Notations}\label{sec:notation}
In this section, we describe our blockchain model and the notation used in the paper.
%To keep it simple and consistent, in the rest of the paper, we will use the term transaction propagation, yet it can be replaced with any kind of information propagation.

%\textbf{Access structure.} 
%There are two types of blockchains: 
%permissionless (public) and permissioned.
%In permissionless blockchain, anyone can join the network and contribute without an authorization. 
%Permissioned blockchain, conversely, has an authority deciding who can contribute to the ledger. 
%%permissionless where anyone can join and contribute to the chain, and permissioned which requires an authorization to join the network.
%Our work does not have a requirement in this manner, and we focus on the generic case: permissionless blockchain.

\noindent\textbf{Network.} 
It is a dynamic peer-to-peer network means that there are nodes joining and leaving constantly.
%The network can also be represented as an undirected graph.
Unlike to the existing works \cite{abraham2016solidus,babaioff2012bitcoin}, we do not have a restriction on the network topology.
% analyze both cases in a generic network model without restricting the incentive mechanism. 
%We may use graph terminology to describe the network topology.
%Each new node is, by default, connected to $N_{con}$ nodes in the network.
%Some nodes may decide to connect more nodes, default value will be used for simple estimations.
%The network is assumed to be a 2-connected graph, i.e., removal of any node would not disconnect the graph or there are multiple paths between any two nodes.

\noindent\textbf{Participants.}  
Each participant is denoted by a node in the network. 
We assume a permissionless blockchain where anyone can participate and contribute to the ledger directly.
Moreover, there is no discrimination between nodes (participants), i.e., they can all be the owner of a transaction and propose a block as a miner (round leader). 
For identification, each node has a public and private key pair and can be validated by his public key.

\noindent\textbf{Consensus and leader election.}
Incentive mechanism defined in Section \ref{sec:incentive} works regardless of the consensus structure. 
Whereas, the routing mechanism requires special treatment, which we call \textit{first-leader-then-block} (\textit{FLTB}) type consensus protocols.

{FLTB} protocols can be defined as the consensus model where the round leader is validated before he proposes the block.
Any leader election mechanism which is independent of the prospective block of that leader can be converted into {FLTB} type.
Examples of the {FLTB} consensus protocols are Proof-of-Work (PoW) based Bitcoin-NG \cite{eyal16bitcoin} and several Proof-of-Stake (PoS) based ones \cite{bentov2016cryptocurrencies,bentov2014proof,ouroboros}.

%Furthermore, there could be multiple round leaders at a round, both mechanisms can be applied in these scenarios. 
% the routing mechanism will still be efficient as long as the leader set compose a relatively small subset of the participants.

The rest of the definitions and notations are listed below:
\begin{itemize}
%	\item \textit{Node}: A participant of the blockchain in the network. %We prefer to call participants as nodes regarding the networking terminology. %We may use the terms node and participant interchangeably. 
	\item \textit{Neighbor nodes}: Directly connected nodes in the network, adjacency in the graph.
	\item \textit{Client}: The source or the sender of a transaction. Client of a transaction $T$, denoted by $c_T$. %Note that client could be a miner at the same time.
	\item \textit{Round Leader}: The legitimate node (participant) who constructs the current block. 
	%For simplicity, we will assume there is only one leader, but our mechanisms can be applied for systems having multiple leaders at the same time. %As long as the number of leaders is not proportional to the network size, it will be significantly bandwidth-efficient.
	\item \textit{Intermediary Node}: A node on the transmission path between the round leader and a client.
	\item $\mathcal{L}^r$: The credential of round leader which validates the round leader of round $r$ and can be verified by all nodes in the network.
	For example, it could be a special hash value in a PoW protocol or the proof of possessing the chosen coin in a  PoS protocol.
	In general, regardless of the consensus mechanism, credentials are linked to the public key of the leader and can be verified by a corresponding signature.
	\item $\pi{(n_i)}$: The probability of node $n_i$ being the round leader, also referred as the capacity of node $n_i$. 
	It corresponds to the mining power in PoW or the stake size in PoS protocols and is assumed to be greater than zero for every node in the network. 
	$\pi(S)$ corresponds to the total capacity of the all nodes in set $S$.
	%\item $ER^T_i$: The expected reward of node $n_i$ from a transaction $T$.
	\item $\mathcal{N}_K^T$: The set of nodes who know (received) the transaction $T$. $\mathcal{N}_K^{n,T}$ presents the set from the point of view of node $n$ (including $n$ itself).
	\item $\mathcal{N}_{NK}^T$: The set of nodes who do not know (received) transaction $T$ yet. $\mathcal{N}_{NK}^{n,T}$ denotes the set from the point of view of node $n$ and includes only the neighbors of $n$.
	%\item $\mathcal{N}_K^{n,T}$: From point of view of node $n$, the set of nodes who know (have) the transaction $T$.
	%\item $\mathcal{N}_{NK}^{n,T}$: From point of view of node $n$, the set of nodes who do not know (have) transaction $T$ yet.
\end{itemize}

\section{Incentive Mechanism}\label{sec:incentive}
We now describe our incentive mechanism.
For the sustainable functioning of a fully decentralized blockchain where the nodes (participants) are able to create new identities and behave according to their incentives, propagation mechanism needs to be Sybil-proof and incentive compatible \cite{peters2016understanding}.
%In a permissionless setting, incentive-compatibility and rational behavior have been already observed \cite{peters2016understanding} and studied in information propagation manner \cite{babaioff2012bitcoin,abraham2016solidus}.

Conventional incentive instrument, namely transaction fee, almost always refers to the reward of the round leader.
Here, we refer transaction fee as it consists of the reward to propagate and to validate transactions.
Thereby, rational nodes are encouraged not only to validate transactions but also to propagate them.
%It can be dependent on size, functional complexity, number of parties involved, and so on. 
How to determine the fee is out of the scope of this paper but we assume that each transaction fee is predefined by either the client or a known function.
We focus on how to automatically allocate the fee among all the contributors of the process.

\noindent\textbf{Fee sharing function (rewarding mechanism).}
The fee sharing function distributes the transaction fee among the propagating nodes and the round leader. 
Note that it is highly probable that the same transaction is received more than once by the round leader (and intermediary nodes) because of the propagation mechanism. 
A rational round leader would choose the one which maximizes his profit.
Like existing works \cite{abraham2016solidus,babaioff2012bitcoin}, the fee sharing function described here deals with the path which is included in the block.
For a transaction (added to the block) with fee $\ff$ and propagation path $P$, the function $\mathcal{F}$ determines the shares of each node involved:
\begin{IEEEeqnarray*}{l}
\mathcal{F}: \{\ff,P\} \longrightarrow \{ \f{|P|}{i} \}_{i=1}^{|P|} \text{ where } \sum_{i=1}^{|P|}\f{|P|}{i}=\ff.
\end{IEEEeqnarray*}
$|P|$ denotes the number of nodes involved in the processing of a transaction with fee $\ff$, where $|P|-1$ of the nodes are in the propagation path between the client and the round leader.
Let $|P|=k$, i.e., length of the propagation path of the transaction is $k$.
Then, $\f{|P|}{i}$ denotes the share of $i^{th}$ node in the propagation path, $\f{k}{k}$ is the share of the round leader and $\sum_{i=1}^{k}\f{k}{i}=\ff$.

%\textbf{Graph topology.}
%%In addition, we also assume that propagation cost of a transaction is small enough, and can be neglected.
%%Latter, we will discuss more on this assumption. 
%The blockchain network is presumed to be a 2-connected graph which means that if any node is removed then it will stay connected, i.e., there are multiple paths between any two nodes.
%This allow us to suppose that there will be competition between different paths from a client to the leader.
%Otherwise, it is shown in Appendix \ref{sec:appendix} that there is no fee sharing function secure against Sybil attacks other than sharing the fee between the first propagator and the round leader.
%Here, we use the same definition of Sybil nodes in \cite{babaioff2012bitcoin}: fake identities sharing the same neighbors with the original node and do not increase connectivity of the network.  

In the rest of the section, we formulate the necessities of the fee sharing function to incentivize propagation of an arbitrary transaction $T$ with fee $\ff$. % among propagating participants and the round leader. 
An ideal incentive compatible propagation mechanism should satisfy the following properties:
\begin{enumerate}
	\item \textit{Sybil-proofness}: An intermediary node, as well as the round leader, should not benefit from introducing Sybil nodes to the network.
	\item \textit{Game theoretically soundness}: A transaction should not be kept among a subset of the network. 
	There should be adequate incentive for rational nodes willing to propagate, thence it will eventually reach to the whole network.
	%\item \textit{Permanence and Fairness}: The sharing decisions in a certain situation should not lead to a disadvantageous position in the future because of the unpredictable actions taken by the shared parties.
\end{enumerate}

%In the rest of the section, we formulate and determine the automated fee sharing function for an arbitrary transaction $T$ with fee $\ff$.

% Here, we investigate the case where $f$ is totally shared among the nodes in the propagation path. 
%However, for some schemes like Bitcoin-NG \cite{eyal16bitcoin}, $f$ is also shared with the next block leader.
%In Appendix \ref{sec:appendix2}, we mention how to extend fee sharing function for these kind of consensus protocols.
By formulating these conditions, we achieve the following theorem (where $C$ is a constant which can be chosen according to the network connectivity):
\begin{theorem}\label{thm:main}
In a 2- or more connected blockchain network, each rational node $n\in \mathcal{N}_{K}^{T}$ with $ \pi(n) < C\cdot\pi(\mathcal{N}_{K}^{n,T})$ propagates transaction $T$ without introducing Sybil nodes, if the transaction fee $\ff$ is shared by the following method:
\begin{IEEEeqnarray}{ll}
\f{k}{i}= \begin{cases}
\ff\cdot C(1-C)^{i-1} &\text{ for } 1\leq i<k, \\
\ff\cdot (1-C)^{k-1} &\text{ for } i=k.
\end{cases}\nonumber
\end{IEEEeqnarray}
%where $\f{k}{i}$ denotes the share of $i^{th}$ node in the propagation path of length $k$, $\f{k}{k}$ is the share of the round leader, $\sum_{i=1}^{k}\f{k}{i}=\ff$.
\end{theorem}
Proof of the theorem is divided into the following sections.
The requirements are formulated in Sections \ref{sec:sybil} and \ref{sec:gametheory}, and the fee sharing function satisfying them is computed in Section \ref{sec:feesharing}.

\subsection{Sybil-Proofness}\label{sec:sybil}

%As can be seen in Figure \ref{fig:sybilnode}, node $n$ adds node $n'$ and thereby own two nodes in propagation paths. 

%\begin{figure}
%	\begin{center}
%		\includegraphics[height=3cm]{figures/sybil.pdf}
%		\caption{Illustration of a Sybil node $n'$ introduced by node $n$.}\label{fig:sybilnode}
%	\end{center}
%\end{figure}

%Well-connected network assumption helps to discuss the fact that there will be competing nodes, and none of them has very critical positioning in the network.
%For the case where the network is weakly-connected, in Appendix \ref{sec:appendix}, we show that the only way secure the incentive function against Sybil attacks is sharing between the first propagating node and the round leader.
%Thus, it is impossible to share among all contributors while securing against Sybil attacks. 

Here, we use the same definition of Sybil nodes in \cite{babaioff2012bitcoin}: fake identities sharing the same neighbors with the original node that do not increase the connectivity of the network. 
Because of the Sybil-proof consensus algorithm, Sybil nodes do not increase the capacity of their owner, i.e., the probability of being the round leader.

We investigate the problem in two different settings: 1-connected networks and the rest.
$k$-connected network means that removal of any $k-1$ nodes does not disconnect the network.
In 1-connected networks, there exists a bridge which is the only connection between two distinct subnetworks.
Though 1-connected network model seems to be unrealistic topology for permissionless blockchains,
it is important to see the intuition behind the non-competition effect. 

\noindent\textbf{1-connected networks.}
In 1-connected networks, there are critical nodes which have special positions in the propagation paths between some node pairs. 
A critical node for a node pair appears in all possible paths between these two nodes. 
The following lemma shows that non-competing advantage of critical nodes makes it impossible to have a Sybil-proof incentive mechanism for 1-connected networks.

\begin{lemma}[Impossibility Lemma]\label{lem:imp}
For 1-connected networks, there is no Sybil-proof and incentive compatible propagation mechanism which rewards every node in the propagation path. 
%	In order to deviate nodes from introducing Sybil nodes in 1-connected networks, processing fee should be shared between the first propagating node and the round leader.
\end{lemma}
\begin{proof}
	Assume that, because of 1-connectedness of the network, a node $n_i$ have a critical position for a transaction $T$, meaning that it is certain he will be included in the propagation path of that transaction.
	If $n_i$ is one side of the bridge combining two distinct subnetworks, $n_i$ can be sure that each transaction coming from its subnetwork and validated in the other one has to pass through $n_i$.
	In Figure \ref{fig:sybil_appendix}, we illustrate the two possible paths of a transaction passing through $n_i$.
	Since the round leader and also intermediary nodes after $n_i$ will receive one of the paths, they do not have any choice but accept the path sent by $n_i$.
	
	\begin{figure}[htb]
		\begin{center}
			\includegraphics[scale=0.5]{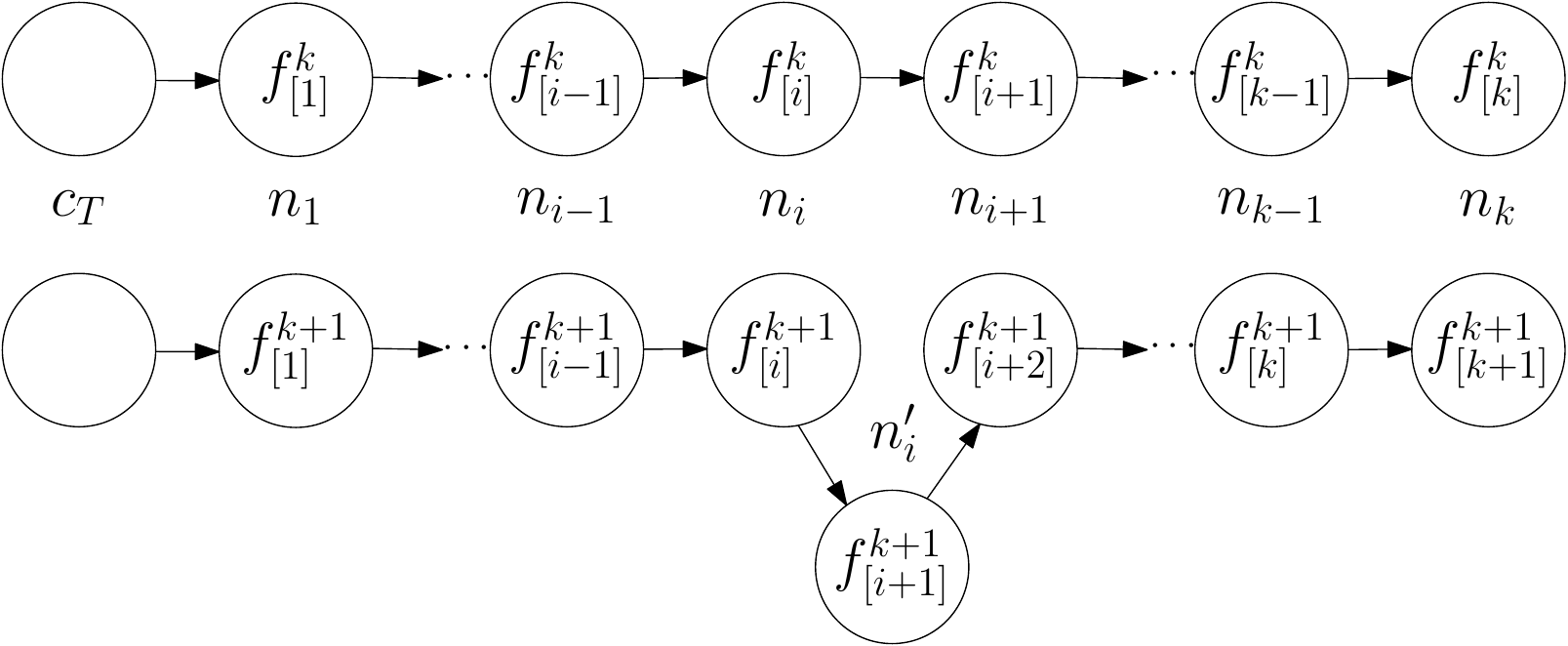}
			\caption{The fee sharing before and after a Sybil node $n_{i'}$ added by $n_i$}\label{fig:sybil_appendix}
		\end{center}
	\end{figure}
	
	Now, we investigate the share of a node $n_i$ with and without a Sybil node. 
	As given in Figure \ref{fig:sybil_appendix}, $n_i$ is the $i^{th}$ node in the original propagation path and his corresponding fee shares are $\f{k}{i}$ and $\f{k+1}{i} +\f{k+1}{i+1}$. 
	In order to demotivate $n_i$, $\f{k}{i}$ should be greater than or equal to $\f{k+1}{i} +\f{k+1}{i+1}$.  
	Since the position of the node would change for different transactions and rounds, the condition should hold for all positions:
	\begin{IEEEeqnarray*}{rll}
		\forall\, i \in \{1,\ldots, k\}, \quad &&\f{k}{i} \geq \f{k+1}{i} + \f{k+1}{i+1} \\
		\text{(summing for all $i$'s)}\implies &&\sum_{i=1}^{k} \f{k}{i} \geq \sum_{i=1}^{k} \f{k+1}{i} + \sum_{i=1}^{k} \f{k+1}{i+1} \\
		\text{(Definition of $\mathcal{F}$)}\implies && \ff \geq \ff - \f{k+1}{k+1} + \ff - \f{k+1}{1} \\
		\implies &&\f{k+1}{k+1} +  \f{k+1}{1} \geq \ff \\
		\text{(Definition of $\mathcal{F}$)} \implies &&\f{k+1}{k+1} +  \f{k+1}{1} = \ff \,.
	\end{IEEEeqnarray*}
	Therefore, other than the first propagating node and the round leader, there is no reward for the rest of the propagating nodes which contradicts with rational behavior. 
\end{proof}

\noindent\textbf{Eclipse and partitioning.} Note that this monopolized behavior is similar to the eclipse and partitioning attacks where the adversary separates the network into two distinct group and controls all the connections between them \cite{heilman2015eclipse,apostolakihijacking}.
Indeed, Lemma \ref{lem:imp} can be generalized to the case where the adversary is able to control all the outgoing connections of a client.
In that case, there is no way to deviate the adversary from creating Sybil nodes for that specific transaction.
We assume that client nodes are able to defend against the eclipse attacks using the countermeasures defined in \cite{heilman2015eclipse}.

In a 2- or more connected network, there are multiple paths between any two nodes. %, including the client and the round leader. 
Therefore, we can immediately focus on the multiple paths case where there are competing paths for the same transaction and the round leader includes one of them to the block. 
%As there are multiple paths/options, nodes can be demotivated from introducing Sybil nodes by following conditions:

A node can profit from a fee by either being an intermediary node who propagates it or being the round leader who creates the block.
We investigate the Sybil-proof conditions of intermediary nodes and the round leader separately.

\paragraph{Intermediary nodes}
An intermediary node can be deviated by the actions of the nodes who receive the transaction afterwards.
Since there are multiple paths, the round leader will receive the same transaction from at least two different paths.
In other words, the round leader would decline all but one of the paths (for each transaction).
An intermediary node will be demotivated if introducing a Sybil node would increase the chance of rejection of his path.

If the share of the round leader decreases as the propagation path length increases, then he will choose the shortest path for each transaction. 
In that case, introducing Sybil nodes will decrease his chance to be included in the block. 
Therefore, providing larger gain to the leader for choosing the shortest path is sufficient and can be formulated as  $\f{k}{k} > \f{k+1}{k+1}$.

\paragraph{Round leader} In some cases, round leader is determined before the block is created or even several rounds earlier \cite{bentov2016cryptocurrencies,bentov2014proof,eyal16bitcoin}.
Since the round leader is guaranteed to be in the propagation path, it is needed to be taken into account separately. 
In addition, an intermediary node can propagate righteously to his neighbors and then add Sybil nodes for his own mining process.
Therefore, in any case (predefined leader or not), it is necessary to make an additional policy for the round leader.

In the case of $s$ Sybil nodes, share of the round leader will change from $\f{k}{k}$ to $\sum_{i=0}^{s} \f{k+s}{k+i}$ for some $k$.
In order to deviate the round leader, $\f{k}{k} \geq \sum_{i=0}^{s} \f{k+s}{k+i}$ is required.

%\begin{itemize}
%	\item \textit{Intermediary nodes}: An intermediary node can be deviated by the actions of the nodes who receive the transaction afterwards.
%	Since there are multiple paths, the round leader will receive the same transaction from at least two different paths.
%	In order words, the round leader would decline all but one of the paths (for each transaction).
%	An intermediary node will be demotivated if introducing a Sybil node would increase the chance of rejection of his path.
%	If share of the round leader decreases as the propagation path length increases, then he will choose the shortest path for each transaction. 
%	In that case, introducing Sybil nodes will decrease his chance to be included in the block. 
%	Therefore, providing larger gain to the leader for choosing the shortest path is sufficient and can be formulated as  $\f{k}{k} > \f{k+1}{k+1}$.
%	\item \textit{Round leader}: In some cases, round leader is determined before the block is created or even several rounds earlier \cite{bentov2016cryptocurrencies,bentov2014proof,eyal16bitcoin}.
%	Since the round leader is guaranteed to be in the propagation path, it is needed to be taken into account separately. 
%	In the case of $s$ Sybil nodes, his share will change from $\f{k}{k}$ to $\sum_{i=0}^{s} \f{k+s}{k+i}$ for some $k$.
%	For Sybil-proofness against the round leader, $\f{k}{k} \geq \sum_{i=0}^{s} \f{k+s}{k+i}$ is required.
%\end{itemize}
Since the latter condition includes the former one (as $\f{k+1}{k}>0$), Sybil proofness condition can be formulated as: 
\begin{IEEEeqnarray}{l}\label{eqn:sybil}
\forall\, k\geq 1, \forall\, s\geq 1, \quad \f{k}{k} \geq \sum_{i=0}^{s} \f{k+s}{k+i}\,.
\end{IEEEeqnarray}

\subsection{Incentive Compatibility}\label{sec:gametheory}

The decision of the propagation of a transaction can be analyzed as a simultaneous move game where each party takes action without knowing strategies of the others. 
All players (nodes in our case) are assumed to be rational and they decide their actions deducing that the others will also act rationally.
Some nodes may cooperate with each other. 
We assume that colluding neighboring nodes already share every transaction with each other and take actions as one. 
In other words, they act as a single combined node in the network which can be seen as Sybil nodes.
% investigated in the previous section.
%From now on, we presume colluding neighboring nodes act as one and deal with none-colluding neighbor nodes.

Here, we investigate the propagation decision by comparing the change in the expected rewards for a transaction $T$.
In the beginning, each transaction is shared with some nodes, at least with the neighbors of the client. 
We will find the required condition to propagate through the whole network.
We first investigate the propagation decision by comparing the change in the expected rewards immediately after the action. 
Then, we extend our analysis with a permanence condition which guarantees that the ones who propagate will not suffer from any future actions.
%In that sense, sharing $T$ with nodes who already have is irrelevant.
%Therefore, sharing with a node implicitly refers to the one who does not have it yet.
%From a node's point of view, a knowing-node ($n\in \mathcal{N}_K^T$) may seem as a not-knowing one ($n\in \mathcal{N}_{NK}^T$). 
%For our analysis, this kind of view may lead to share with a node who already has, but it does not obstruct propagation analysis since a not-knowing node will be seen as not knowing from all his neighbors.

We show that the sharing decision of a node is independent of the probability of his neighboring nodes being the round leader.
Instead, it depends on his own probability against the rest who knows the transaction.

\begin{lemma}[Equity Lemma]\label{lem:share}
Propagation decision of a node is independent from the neighbors' capacities. A rational node would propagate to either all of its neighbors or none of them.
\end{lemma}
\begin{proof}
Let a transaction $T$ with fee $\ff$ is known by a node $n$, and its distance to the $c_T$ is $k$. 
The expected reward of node $n$ can be defined as a function $R(\cdot)$ whose input corresponds to the capacities of the nodes who received $T$ from $n$, then
\begin{IEEEeqnarray*}{l}
R(X)=\frac{ \f{k}{k}\cdot\pi(n) +\f{k+1}{k}\cdot X }{ \pi({\mathcal{N}_{K}^{n,T}}) + X }.
\end{IEEEeqnarray*}

We show that $R(\cdot)$ is a monotone function.
In order to show that a function is a monotone, it is enough to show that the sign of its derivative does not change in the domain range. 
For our case, it can be seen that the sign is independent of the input:
\begin{IEEEeqnarray*}{lll}
R'(X)&=&\frac{ \f{k+1}{k}\left(\pi({\mathcal{N}_{K}^{n,T}}) + X \right) - \left(\f{k}{k} \pi(n) +\f{k+1}{k} X \right)}{ \left(\pi({\mathcal{N}_{K}^{n,T}}) + X \right)^2 }\\
&=&\frac{  \f{k+1}{k} \pi({\mathcal{N}_{K}^{n,T}}) -\f{k}{k} \pi(n) }{  \left(\pi({\mathcal{N}_{K}^{n,T}}) + X \right)^2 } .
\end{IEEEeqnarray*}

Since $R(\cdot)$ is a monotone function, then it achieves the maximum value at one of the boundary values. 
In our case, the boundary values are $X=0$ where no neighbors received the transaction and $X=\pi \left( \mathcal{N}_{NK}^{n,T} \right)$ where all neighbors received it.
Here, we omit the fact that $\pi(\cdot)$ is also a monotone function.
Thus, we can say that a rational node maximizes his profit by propagating to either all of its neighbors or none of them.
\end{proof}

Lemma \ref{lem:share} simplifies to evaluate interfering multiple node decisions which is discussed in the following Lemma.

\begin{lemma}[Propagation Lemma]\label{lem:ER}
Let a node $n\in \mathcal{N}_K^T$, $\mathcal{N}_{NK}^{n,T}\neq \emptyset$ where the distance between $n$ and $c_T$ is $k$. All neighbors of $n$ will be aware of $T$ if 
\begin{IEEEeqnarray*}{l}
\frac{\f{k+1}{k}}{\f{k}{k}} > \frac{\pi(n)}{\pi({\mathcal{N}_{K}^{n,T}})}\,.
\end{IEEEeqnarray*}
\end{lemma}
\begin{proof}
Assume that some of the neighbors of $n$ are not aware of $T$, i.e., $\mathcal{N}_{NK}^{n,T}\neq \emptyset$. 
From Lemma \ref{lem:share}, we know that $n$ did not propagate the transaction to any of his neighbors. 
Therefore, at the moment, the only way that $n$ profits from $T$ is being the round leader with a reward $\f{k}{k}$. 

\begin{table*}[htbp!]
\caption{The expected reward of $n$ from $T$ regarding possible decisions of $n$ and the rest of  $\mathcal{N}_{K}^{n,T}$.}\label{tab:ER}
\begin{center}
\begin{tabular}{ccccc}
	&  & \multicolumn{3}{c}{$\mathcal{N}_{K}^{n,T}$ (excluding $n$)} \\ %\cline{3-5} 
	%\cline{3-4} 
   	& \multicolumn{1}{c|}{Decision} & \multicolumn{1}{c}{Not Propagate} & & \multicolumn{1}{c}{(some) Propagate} \\ %[0.5cm]
	\cline{2-5} 
	%\hline 
	\multirow{4}{*}{$n$}	& \multicolumn{1}{c|}{Not Propagate} & $\frac{\f{k}{k}\cdot\pi(n)}{ \pi({\mathcal{N}_{K}^{n,T}}) }$  & & \multicolumn{1}{c}{$ \frac{\f{k}{k}\cdot\pi(n)}{ \pi({\mathcal{N}_{K}^{n,T}}) + \pi({CN}) + \pi({NCN_2}) }$ } \\
	%\cline{2-2}
	& \multicolumn{1}{c|}{} & & \\
	& \multicolumn{1}{c|}{Propagate} & $\frac{\f{k}{k}\cdot\pi(n)+\f{k+1}{k}\cdot \pi({\mathcal{N}_{NK}^{n,T}}) }{ \pi({\mathcal{N}_{K}^{n,T}}) + \pi({\mathcal{N}_{NK}^{n,T}}) }$ & &  $\frac{\f{k}{k}\cdot\pi(n)+\f{k+1}{k}\cdot \pi(NCN_1)+ \alpha \f{k+1}{k}\cdot \pi(CN)  }{ \pi({\mathcal{N}_{K}^{n,T}}) + \pi({\mathcal{N}_{NK}^{n,T}}) + \pi({NCN_2})}$\\ 
	%\cline{2-4}
\end{tabular}
\end{center}

\end{table*} 

Table \ref{tab:ER} presents expected reward of $n$ with respect to each possible action of $n$ and $\mathcal{N}_{K}^{n,T}$.
The propagation decision of $\mathcal{N}_{K}^{n,T}$ may not include all its members, thereby all possible decisions are taken into account. 
Here, $CN$ corresponds to the common neighbors of $n$ and $\mathcal{N}_{K}^{n,T}$, $NCN_1$ distinct neighbors of $n$ and $NCN_2$ distinct neighbors of $\mathcal{N}_{K}^{n,T}$ (who decide to propagate), i.e., $CN\bigcup NCN_1 = \mathcal{N}_{NK}^{n,T}$. 
Since $CN$ is received the transaction from both $n$ and the rest of the $\mathcal{N}_{K}^{n,T}$, $\alpha$ represents the percentage of the ones in $CN$ decided to continue with the one including $n$.

If all nodes of $\mathcal{N}_{K}^{n,T}$ decide not to propagate with their neighbors, then $n$ will benefit from propagating $T$ in the case of 
\begin{IEEEeqnarray*}{l}
\resizebox{0.48\textwidth}{!}{$
 \frac{\f{k}{k}\cdot\pi(n)+\f{k+1}{k}\cdot \pi({\mathcal{N}_{NK}^{n,T}}) }{ \pi({\mathcal{N}_{K}^{n,T}}) + \pi({\mathcal{N}_{NK}^{n,T}}) } > \frac{\f{k}{k}\cdot\pi(n)}{ \pi({\mathcal{N}_{K}^{n,T}}) } 
 \iff \frac{\f{k+1}{k}}{\f{k}{k}} > \frac{\pi(n)}{\pi({\mathcal{N}_{K}^{n,T}})}.$}
\end{IEEEeqnarray*}

If (some) nodes in $\mathcal{N}_{K}^{n,T}$ decide to propagate $T$, then $n$ will benefit from propagating $T$ in the case of 
\begin{IEEEeqnarray*}{l}
\resizebox{0.48\textwidth}{!}{$	  \frac{\f{k}{k}\cdot\pi(n)+\f{k+1}{k}\cdot \pi(NCN_1)+ \alpha \f{k+1}{k}\cdot \pi(CN)  }{ \pi({\mathcal{N}_{K}^{n,T}}) + \pi({\mathcal{N}_{NK}^{n,T}}) + \pi({NCN_2})} 
	> \frac{\f{k}{k}\cdot\pi(n)}{ \pi({\mathcal{N}_{K}^{n,T}}) + \pi({CN}) + \pi({NCN_2}) } $}\\
\resizebox{0.48\textwidth}{!}{	$ \impliedby \frac{\f{k+1}{k}}{\f{k}{k}} > \frac{\pi(n)}{ \pi({\mathcal{N}_{K}^{n,T}}) + \pi({CN}) + \pi({NCN_2}) } \text{ and } NCN_1\neq\emptyset.$}
\end{IEEEeqnarray*}

Note that $NCN_1 = \emptyset$ means that all the neighbors of $n$ are also neighbors of $\mathcal{N}_{K}^{n,T}$ who decide to propagate.
In addition, the sufficiency condition is independent of $\alpha$.
Therefore, in any case, if $\frac{\f{k+1}{k}}{\f{k}{k}} > \frac{\pi(n)}{\pi({\mathcal{N}_{K}^{n,T}})}$ is satisfied, then all neighbors of $n$ will be aware of the transaction.
\end{proof}

%Lemma \ref{lem:share} and \ref{lem:ER} states the fact that sharing a transaction is independent from neighbor's probability of being round leader, instead it depends on the his own probability against the rest.

\begin{corollary}\label{corollary1}
Let $\f{k+1}{k}\geq C\cdot \f{k}{k}$ for some constant $C\in(0,1)$. $\mathcal{N}_{K}^{T}$ will continue to expand until there is no more node $n\in \mathcal{N}_{K}^{T}$ having neighbors in $\mathcal{N}_{NK}^{T}$ and satisfying $ \pi(n) < C\cdot\pi(\mathcal{N}_{K}^{n,T})$.
\end{corollary}
%\begin{proof}
%Proof follows from Lemma \ref{lem:ER}.\qed
%\end{proof}

\noindent\textbf{Remark I.} 
Here, it is possible to define different $C_k$ values for each distance $k$, i.e., $\f{k+1}{k}\geq C_k\cdot \f{k}{k}$. %since it is reasonable to assume that with the distance increased, the competition is also scaled up. 
One might argue that, as the distance increases, it could be possible to find nodes satisfying $ \frac{ \pi(n) }{ \pi(\mathcal{N}_{K}^{n,T}) } < C_k$ for smaller $C_k$ values.
However, as seen in Section \ref{sec:combined}, this is not always the case. 
In addition, the intermediate node may not know the exact distance, thus using the same $C$ value would make the decision simpler.

\noindent\textbf{Remark II.} 
Note that the propagation decision is based on $\mathcal{N}_{K}^{n,T}$ instead of $\mathcal{N}_{K}^{T}$ since the latter one may not be available.
This could lead to better consequences for propagation because nodes may predict $\mathcal{N}_{K}^{T}$ greater than its actual size and decide accordingly.
Nonetheless, a carefully chosen $C$ value will lead the nodes to share it with an overwhelming probability.

\noindent\textbf{Remark III.}
Being the round leader should be more appealing than being an intermediary node, thus the round leader would try to fulfill the round block capacity to maximize his profit. 
The system may not work at full capacity if the nodes gain the same reward from propagating instead of validating (as the round leader) transactions.
%In addition, unlike to conventional transaction fee given to the round leader, processing fee is shared with propagating nodes.
In Corollary \ref{corollary1}, the propagation condition is given as $\f{k+1}{k}\geq C\cdot \f{k}{k}$.
We fix the condition in favor of the round leader:
\begin{IEEEeqnarray}{l}\label{eqn:incentive1}
\forall\, k,\quad \f{k+1}{k}= C\cdot \f{k}{k}\,.
\end{IEEEeqnarray}

\noindent\textbf{Permanence condition.} 
In the simultaneous move analysis, we investigated one step at a time, i.e., what will happen immediately after the decision of propagation. 
However, all possible future actions should be taken into account.
For example, the sender of a transaction should consider the possibility of the further propagation done by the receiver. 
From Lemma \ref{lem:share}, capacities of the neighboring nodes do not have any influence on the sharing decision. 
Unless the processing fee share decreases, which is caused by some possible future actions like increased path length, the same lemma will be satisfied.  
If the share of a propagating node is non-decreasing with respect to the path length, then the ones who propagate will not suffer from any future actions. This can be formulated as
\begin{IEEEeqnarray}{l}\label{eqn:incentive2}
\forall\, i<k, \quad \f{k}{i}\geq\f{k+1}{i} \,.
\end{IEEEeqnarray} 

\subsection{Fee Sharing Function}\label{sec:feesharing}
With the equations obtained from the required conditions, we can uniquely determine the fee sharing function and conclude Theorem \ref{thm:main}.
First, using permanence condition (\ref{eqn:incentive2}), Sybil-proofness condition (\ref{eqn:sybil}), can be reduced to $ \f{k}{k} \geq \f{k+1}{k+1} + \f{k+1}{k}$:
\begin{IEEEeqnarray*}{lll}
\forall\, k\geq 1,\;  \f{k}{k} &&\geq \f{k+1}{k+1}+\f{k+1}{k}\geq \f{k+2}{k+2}+\f{k+2}{k+1}+\f{k+1}{k}\\
&&\geq \f{k+3}{k+3}+\f{k+3}{k+2}+\f{k+2}{k+1}+\f{k+1}{k}\geq \cdots\\
\forall\, s\geq 1,&&\geq \f{k+s}{k+s}+\sum_{i=0}^{s-1}\f{k+i+1}{k+i}
\geq \f{k+s}{k+s}+\sum_{i=0}^{s-1}\f{k+s}{k+i}. 
%&&= \sum_{i=0}^{s} \f{k+s}{k+i} \,.
\end{IEEEeqnarray*}
Therefore, we can update the Sybil-proofness condition as:
\begin{IEEEeqnarray}{l}
\forall\, k\geq 1, \quad \f{k}{k} \geq \f{k+1}{k+1} + \f{k+1}{k} \,.\label{eqn:sybil_new}
\end{IEEEeqnarray}
Then, we can obtain the following equations:
\begin{IEEEeqnarray}{lll}
\text{Using (\ref{eqn:sybil_new})} & & \sum_{i=1}^{k} \f{i}{i} \geq \sum_{i=1}^{k} \f{i+1}{i+1} + \sum_{i=1}^{k} \f{i+1}{i}\nonumber\\
										&\implies& \ff=\f{1}{1} \geq \f{k+1}{k+1}+ \sum_{i=1}^{k} \f{i+1}{i} \nonumber\\
\text{Using (\ref{eqn:incentive2})}			&\implies& \ff \geq \f{k+1}{k+1}+ \sum_{i=1}^{k} \f{k+1}{i} = \ff \nonumber\\
										&\implies& \f{k}{i} = \f{k+1}{i} \text{ and } \f{k}{k}=\f{k+1}{k+1}+\f{k+1}{k}. \label{eqn:equal}
\end{IEEEeqnarray}			
After all, we can finalize the fee sharing function which corresponds to Theorem \ref{thm:main}. Using 
(\ref{eqn:incentive1}) and (\ref{eqn:equal}), the share of the round leader can be computed:
\begin{IEEEeqnarray}{l}
 \f{k}{k}=\f{k-1}{k-1}(1-C)=\cdots= \ff\cdot (1-C)^{k-1}. \label{eqn:lead}
\end{IEEEeqnarray}
Using (\ref{eqn:equal}) and (\ref{eqn:lead}), the share of an intermediary node can be computed:
\begin{IEEEeqnarray}{l}
\forall\, i<k,\quad  \f{k}{i} = \f{i+1}{i} = \ff\cdot C(1-C)^{i-1}\,. \nonumber
\end{IEEEeqnarray}

%\subsection{Experimental Results}
%\todo{For Bitcoin network, average share of the contributors are given in table ...}

\subsection{Discussion}\label{sec:inc_discussion} 

\noindent\textbf{Integration.} 
Implementation of the incentive mechanism should take into account the security and efficiency concerns.
The propagation path should be immutable in a way that an adversary cannot add or subtract any node neither in the propagation process nor after the block generation.
At the same time, storage efficiency is also essential since these path logs are needed to be stored in the ledger by every node.
Both existing incentive-compatible blockchain solutions \cite{babaioff2012bitcoin,abraham2016solidus} adopted a signature chaining mechanism where each propagated message includes the public key of the receiver and signature of the sender.
This protocol prevents any manipulation over the path and thereby secures the shares of each contributor. 
It requires additional storage which is the signatures of the contributors.
Although signature chaining solution requires the knowledge of the public key of the receiver and stores signatures of each sender, it is generic and can be applied to any blockchain.
In Section \ref{sec:combined}, we present a novel and storage-efficient solution which is feasible for \textit{FLTB} blockchains.
It is embedded into routing mechanism and does not require the knowledge of the public keys of the neighboring nodes.

\noindent\textbf{Determining $C$ parameter.} 
$C$ value plays an important role to make sure that there will be incentive to propagate a transaction for some nodes until it reaches to the whole network.
On the one hand, as the choice for the $C$ value increases, it will be easier to satisfy the propagation condition since there will be more chance to find nodes satisfying $\pi(n) < C\cdot\pi(\mathcal{N}_{K}^{T})$.
On the other hand, the higher $C$ value, the lower fee remains for the rest of the propagation path. 
It significantly reduces the fee of the round leader, thereby the incentive.
%In addition, it does not help with the assumption about transaction fee being greater than transmission cost.
For these reasons, it is required to choose a moderate $C$ value, e.g., a reasonable choice would be $C=\frac{2}{N_{con}}$ where $N_{con}$ denotes default number of connections of a node.
For example, in Bitcoin network where $N_{con}=8$, nodes will propagate unless they assume that their mining power is greater than 25\% of the ones having the transaction.
Even at the very beginning, at least $N_{con}$ nodes have the transaction, $C=\frac{2}{N_{con}}$ setting would provide overwhelming probability to have nodes willing to propagate according to Corollary \ref{corollary1}.
%More specifically, unless more than half of the nodes colluding, there will be nodes whose probability of being round leader is less than $\frac{2}{N_{con}}$.

\noindent\textbf{Client ($0-$capacity) nodes.} 
The main goal of the propagation incentive mechanism is to make sure that the transactions are received by the nodes who are capable of validating transactions as well as creating blocks. 
For that reason, we mainly focused on the nodes having a capacity greater than zero, i.e., $\pi(\cdot)>0$.
Nevertheless, a client node can be seen as a potential capacity node because of the possible propagation of the client.
Regarding Lemma \ref{lem:share} and permanence condition (\ref{eqn:incentive2}), a rational node, who decided to propagate, would benefit from propagating to the client nodes as well.
At the same time, a client node will always benefit from propagating any transaction since otherwise it will not have any chance to gain a fee.

\noindent\textbf{Decentralization effect.}
In the conventional permissionless blockchains, all rewards including block reward and transaction fees are given to the block owner. 
In other words, nodes have only one incentive to participate in the network: being round leader.
The less chance individual nodes have to be the round leader, the more they are motivated to join into centralized forms (e.g. mining pools) \cite{gervais2014bitcoin,Lewenberg:2015:BMP:2772879.2773270}.
Conversely, the transaction fee is shared with all propagators nodes.
In addition, since many transactions are included in a single block, aiming processing fees of (some) transactions has significantly more chance than being the round leader.
Thereby, it is reasonable to conclude that incentive mechanism would have a positive impact on the decentralization of the permissionless blockchains.
\section{Routing Mechanism}\label{sec:routing}

As a non-hierarchical peer-to-peer network, the blockchain ledger is validated by all nodes (miners) individually.
This requires broadcasting every data and blocks over the network since every node needs to keep a record of the chain to validate new blocks.
%In order to validate blocks, every node is required to store previous ones, thereby it is inevitable to broadcast proposed blocks including transactions. 
In existing permissionless blockchains, every transaction is broadcast throughout the network by the client, then the new block including (some of) these is constructed and broadcast by the round leader.
Hence, each transaction is broadcast at least twice.
Even more (\texttt{inv}) messages are sent to check the awareness of the neighbors on the transaction.

In Nakamoto-like consensus protocols, the round leader is validated simultaneously with his proposed block where the redundant propagation of the client is inevitable.
In \textit{FLTB} protocols, on the other hand, it is possible to validate the round leader before the block is proposed.
It enables to determine a direct route between each client and the round leader.
Our routing mechanism in Algorithm \ref{alg:routing_v2} finds the shortest paths between clients and the round leader for each round.
Instead of sending each transaction to all nodes in the network, it is relayed over the shortest path between the client and the leader. 
The distance between (almost) any two nodes in a connected graph is dramatically smaller than the size of the network \cite{travers1967small}.
This is equivalent to cost reduction from $O(N)$ to $O({\ln N})$ in a random network of size $N$ \cite{albert2002statistical,erdos1960evolution}.

The treat model of routing mechanism we present in this section considers a malicious adversary rather than a rational one. 
In the routing mechanism, a malicious adversary may try to block or censor some of the transaction propagations.

Our protocol can be divided into two parts: \textit{Recognition Phase} where the routes are determined and \textit{Transaction Phase} where the transactions are propagated (see Figure \ref{fig:routing}). 
First, in the recognition phase, the round leader is recognized throughout the network and his credential is propagated with a standard gossip protocol. 
Each node $n_i$ learns his closest node towards the round leader, \textit{gradient node} ($gn_i$), who is the first node forwarding the credential.
In the transaction phase, each client forwards his transaction to (some of) his neighbors. 
Then, each node, receiving a transaction for the first time, directly transmits to his gradient node.
Here, the reason for clients to broadcast to more than one neighbor is that one path could yield a single point of failure.
It could be caused by the nodes who fail or maliciously censor some of the transactions.
%Even though it is reasonable to assume that the active nodes in the first phase would be more likely to stay active, it is better to send transactions more than one neighboring node to avoid failure to reach the round leader.
As presented in the experimental results, forwarding transaction to a few of the neighbors (precisely $N_{con}$) is sufficient.
Note that, the routing mechanism works under asynchronous network assumptions since a client does not have to wait for all nodes but $N_{con}$ of his neighbors.
Similarly, for an intermediary node, waiting for the first credential message is enough to propagate received transactions.

\alglanguage{pseudocode}
\begin{algorithm}[htp!]
	\begin{algorithmic}
		%\Require Leader credential $\mathcal{L}^r$
		%\Ensure For each node, closest node to the leader $cn_i$, for $i=1,\ldots,n$
		\State
		\State \underline{Recognition Phase}
		\State Leader provides his credential $\mathcal{L}^r$ to his neighbors.
		\For{Node $n_1$ to $n_N$}
		\If{First time receiving $\mathcal{L}^r$}
		\State Store ID of the sender (gradient) node $n_j$, i.e., $gn_i \leftarrow n_j$
		\State Propagate $\mathcal{L}^r$ to neighbors.
		\EndIf
		\EndFor
		\State 
		\State \underline{Transaction Phase}
		\State  Client provides transaction $T$ to his neighbors.
		\For{Each node $n_i$ receiving $T$}
		\If{First time receiving $T$}
		%\State Add his signature to the chain of signatures \todo{HOW??}
		\State Send it to the $gn_i$
		\EndIf
		\EndFor
	\end{algorithmic}
	\caption{The Routing Algorithm} 
	\label{alg:routing_v2}
\end{algorithm}

\begin{figure}[htp!]
	\begin{center}
		\includegraphics[scale=0.4]{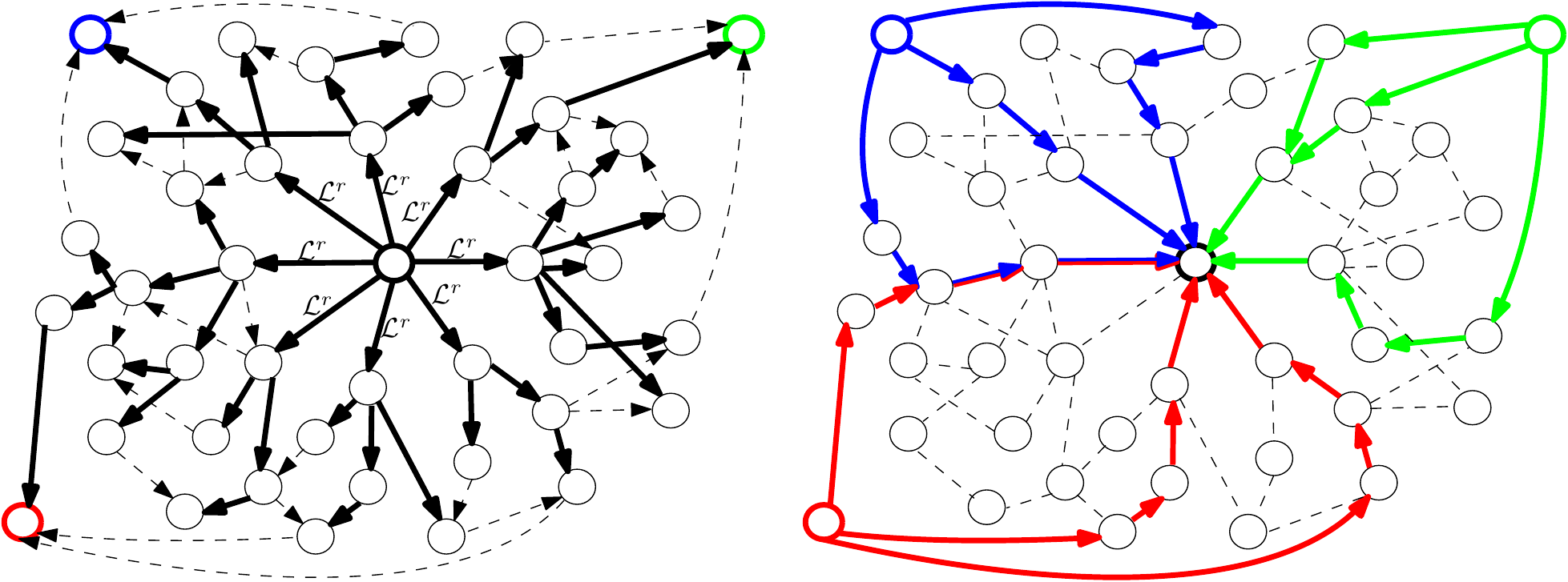}
		\caption{The Routing Mechanism. The left one illustrates the Recognition Phase and connections to the gradient nodes are shown with bold solid lines. On the right, three clients and their transaction paths are presented.}\label{fig:routing}
	\end{center}
\end{figure}

\noindent\textbf{Locational privacy.} 
There have been several papers investigating anonymity in the permissionless blockchain networks, especially for the Bitcoin network \cite{biryukov2014deanonymisation,NIPS2017_6735,Koshy2014}. 
It is found out that matching public keys and IP addresses can be done by eavesdropping.
In this manner, \textit{FLTB}-based blockchains may expose to DoS (denial-of-service) attacks against to the round leader.
We want to stress that our routing mechanism does not leak any more locational information about the position of the leader other than the original \textit{FLTB} protocols do.
It just takes advantage of the announcement of the leader which is done exactly in the same manner with the \textit{FLTB} protocols.
%Therefore, with routing scheme, it enjoys bandwidth efficiency while keeping the same security level for the round leader like DDoS attacks.
Therefore, our routing mechanism does not cause any additional vulnerabilities for DoS-like attacks against the round leader.
Yet, it is possible to improve the locational privacy via anonymity phase where the message is first forwarded in a line of nodes, then diffused from there \cite{bojja2017dandelion}. 
The extra cost of anonymity would be a few nodes on the line which is still proportional to the logarithmic size of the network. 

%\paragraph{Multiple Round Leaders.} For some consensus protocols, there may be multiple parties involved in the block generation.
%If all parties have right to propose at the same time interval, then they can simultaneously run recognition phase. Intermediary nodes may have a corresponding gradient node for each potential leader, and propagate transactions for each direction.
%If parties have ordered individual proposing intervals, then either they may run recognition phase one-by-one or all at the beginning depending on the size of the set.
%Nonetheless, the routing mechanism will remain efficient while the number of transactions are significantly larger than the size of set of round leaders.

\subsection{Experimental Results}
%\textbf{Network Topology} 
%In this experiment, we did not mimic network approximations of the existing blockchains applications since they do not represent a fully decentralized blockchain structure \cite{miller2015discovering,gencer2018decentralization}.
%Also, it is not possible to fully visualize their networks because of the mitigations against DOS attacks.
In this experiment, we use Barab{\'a}si-Albert (BA) graph model \cite{albert2002statistical} which simulates peer discovery in a peer-to-peer network. 
It starts with a well-connected small graph and each new node is connected to some of the previous nodes with a probability proportional to their degrees.

Barab{\'a}si-Albert (BA) \cite{albert2002statistical} and Erd{\H{o}}s-R{\'e}nyi (ER) \cite{erdos1960evolution} graph models have been used to simulate permissionless blockchains \cite{neudecker2016timing,berini2015bitcoin}.
In our setting, we combine both models where the network starts with a small ER graph and grows according to BA model.
We start with 50 nodes in ER model \cite{erdos1960evolution} with edge probability of $1/2$, meaning that on average each node has 25 connections. 
Then, each new node is added by connecting with $N_{con}$ nodes in the network.
For each $(N,N_{con})$ pair analyzed, we generated various graphs using Python graph library \cite{igraph}. 
%We mainly illustrates the results for $N_{con}=8$ which corresponds to Bitcoin default value.

\noindent\textbf{Bandwidth gain.} In \cite{randomgraphs}, the average shortest path length between any two nodes, i.e., the average path length, of a BA graph is shown to be in the order of $\frac{\ln N}{\ln \ln N}$. Hence, our routing protocol reduces the communication cost of a message transaction from $O(N)$ to $O(N_{con}\cdot \frac{\ln N}{\ln \ln N})$. 
The communication gain is up to $99\%$ for scaled networks (see Figure \ref{fig:gain}), which can be verified by counting the average number of nodes visited per transaction. 
Here, we assume that the first arriving credential is coming from the node which is closest to the leader with respect to the number of nodes in between. In other words, the delay between any two nodes is computed by the node-distance.

%\textbf{Flooding.} 
In Figure \ref{fig:gain}, we count only one redundant communication for each transaction. Even more redundancy is caused by the flooding of each transaction because the same transaction is received from different neighboring nodes.
In other words, the total redundancy is not $N$, but on average $N_{con}\cdot N$.
In the existing blockchains, this additional redundancy is reduced by the sending the hash of the transaction to check whether the neighbor has it or not.
If storage size of the transaction is relative to the size of the hash, then the total number of relays of a transaction would be significantly more than double of the network size.
For example, Statoshi info \cite{statoshi}, a block explorer of Bitcoin, shows that average incoming bandwidth usage for the transactions (\texttt{tx}) is, $2.87$ KBps, less than for the checking messages (\texttt{inv}), $4.12$ KBps (measurements taken between 02:00 AM and 14:00 PM in 13 of Feb. 2018).
To conclude, since our mechanism does not suffer from the flooding effect, the actual communication gain would be much higher than the result in Figure \ref{fig:gain}.

\begin{figure}[!htb]
	%\centering
	\includegraphics[scale=0.18]{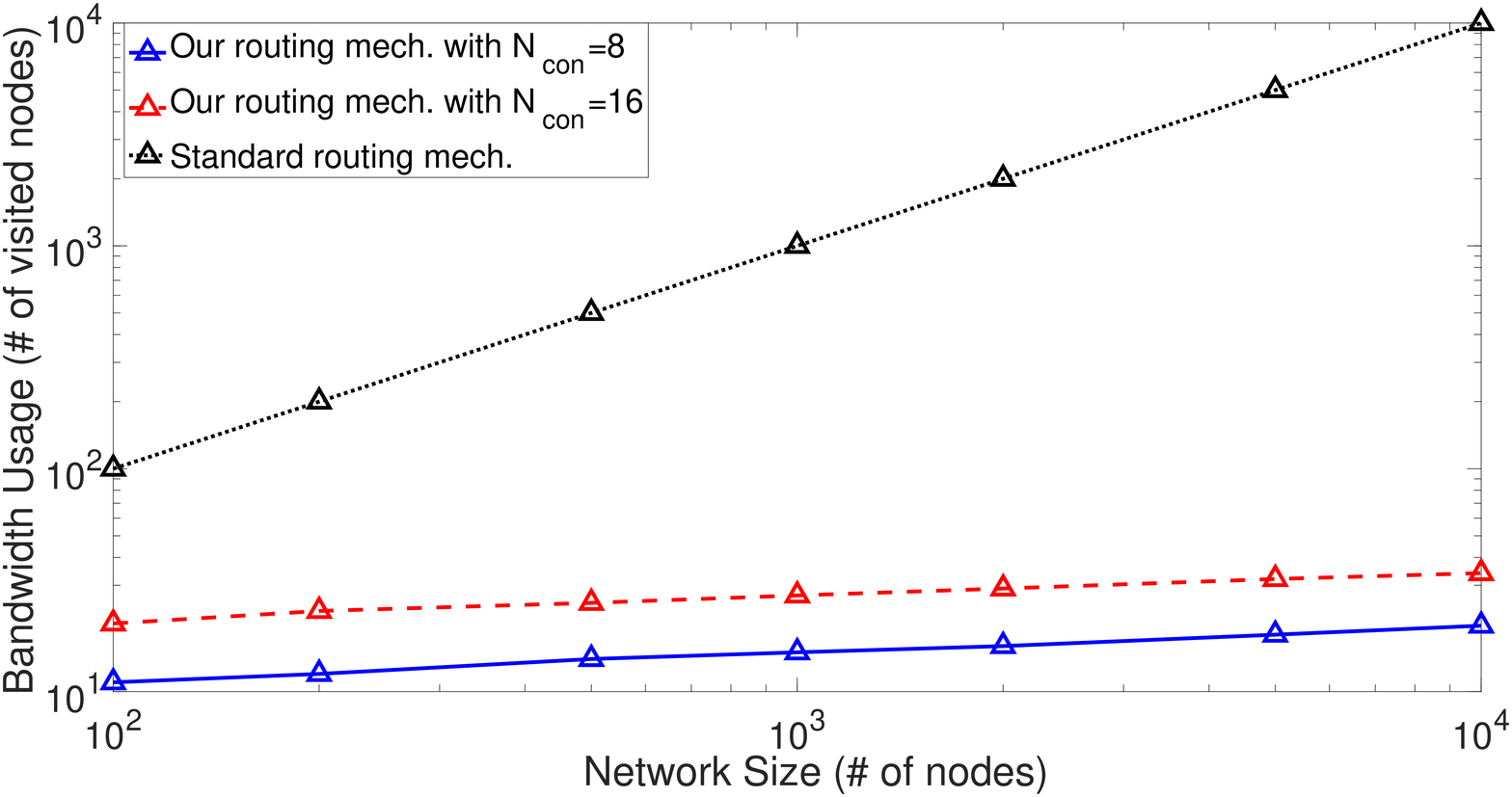}
	\caption{Communication cost for advertisement of a transaction.}
	\label{fig:gain}
\end{figure}

\noindent\textbf{Failing transmissions.} Since each transaction is propagated among a small set of nodes, we need to take into account the possibility of propagation failure which can be caused by the nodes who fail or censor the transaction. 
The failure probability of a transaction can be approximated by $\left( 1-(1- h)^{\frac{\ln N}{\ln \ln N} -1}\right)^{N_{con}}$ where $h$ denotes the probability of a node in the network who fails or censors the transaction. 
These failing nodes are the ones who were present at the recognition phase and failed just afterwards.
Long-term offline nodes can be ignored since they will not be chosen as gradient nodes.
Thus, Figure \ref{fig:failing} demonstrates that our routing is robust against instant network fluctuations.
%Table \ref{tab:fail} shows that the percentage of the failing transactions is nearly negligible and can be reduced even more by increasing the number of paths.  
For a blockchain network with $N=10000$ and $N_{con}=8$, similar to Bitcoin network, if $30\%$ of the active nodes fail after the recognition phase, only $9\%$ of the transactions will be affected.

\begin{figure}[!htb]
	%\centering
	\includegraphics[scale=0.18]{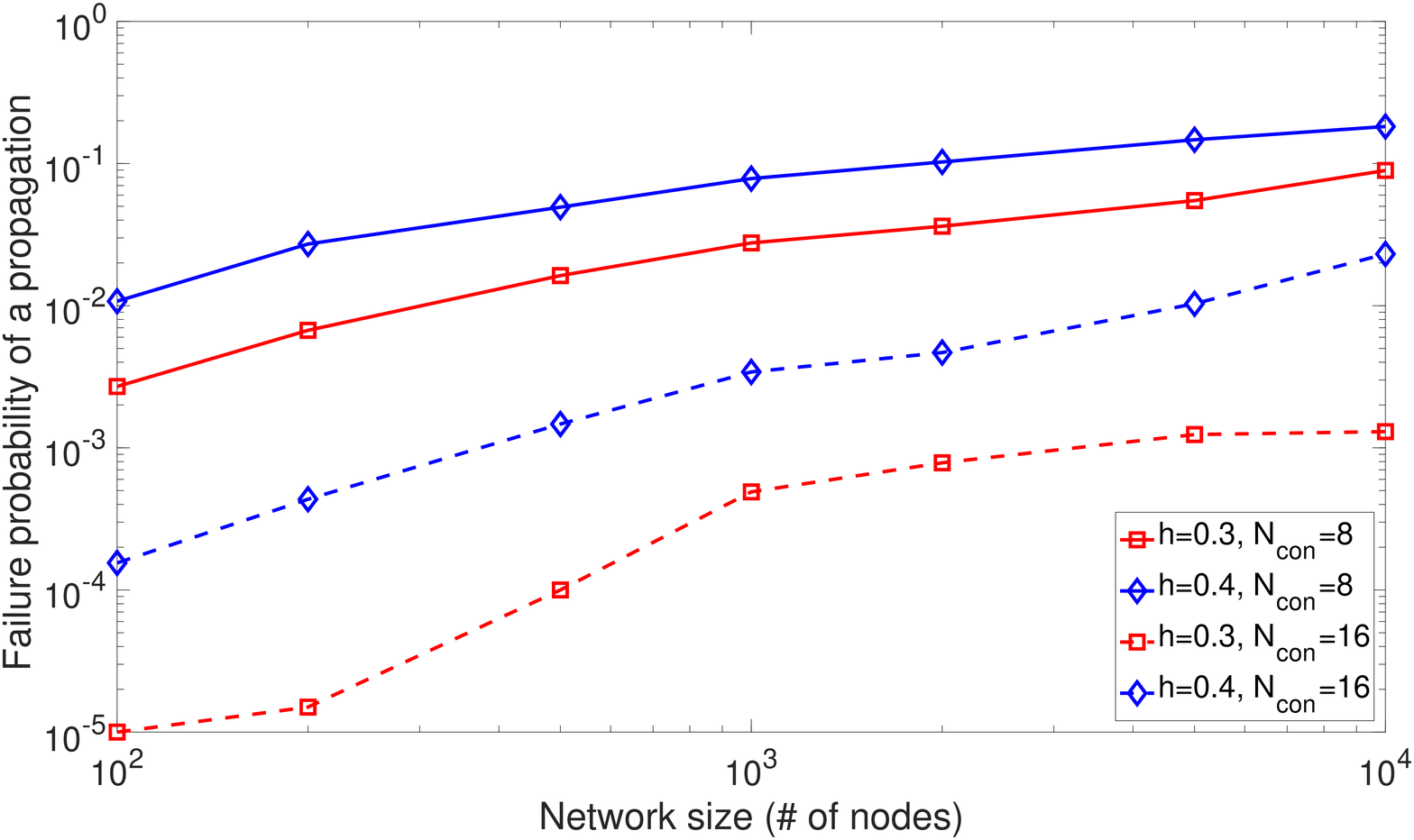}
	\caption{Probability of a transaction failing to be received by the round leader where $h$ is the probability of an intermediary node being a failing or censoring node.}
	\label{fig:failing}
\end{figure}
\section{Combined Propagation Mechanism}\label{sec:combined}
In this section, we show how to deploy both of the incentive and routing mechanisms for any blockchain having a \textit{FLTB} consensus protocol. 
At first glance, they seem to conflict with each other because the incentive mechanism is used to encourage propagation while the routing mechanism helps to reduce redundant propagation.
We combine them in a way that rational nodes are encouraged to propagate only the transactions which are coming from the predefined paths of the routing mechanism.
%A naive approach implementation of both mechanisms may not be compatible with each other.
%For example, because of the propagation incentive, nodes may propagate to as many nodes as possible to increase their chance, and that would conflict with the advantage of smart routing.
As demonstrated in Algorithm \ref{alg:combined}, we use the same infrastructure with the routing mechanism, and we include proofs of the intermediary nodes such that their contributions cannot be denied. 
Each transaction path is defined and secured by a path identifier which includes the public keys of the propagating nodes.
Blocks consist of transactions as well as their path identifiers used to claim processing fee shares.

In the recognition phase, each intermediary node conveys the leader credential and the path identifier.
Incoming and outgoing path identifiers of a node $n$ are denoted by $IN_{n}$ and $OUT_{n}$, which are used to validate and secure the propagation path.
The round leader $\ell$ produces the initial identifier, $OUT_{\ell}=H(\mathcal{L}^r,PK_{\ell})$, and propagates to his neighbors.
Each node $n$ updates the identifier coming from the gradient node by $OUT_{n}=H(IN_{n},PK_{n})$.
This operation is done just for the gradient node (first one sending $\mathcal{L}^r$), then updated identifier and the credential are forwarded to the neighbors.
Nodes may ignore the subsequent identifiers except a client who stores the first $N_{con}$ ones for the transaction phase.

After the routing paths are determined, each client delivers the signed transaction and the incoming identifier to his $N_{con}$ neighbors.
The first receiving nodes, check the signature, then add their public keys to the transaction and forward it to their gradient nodes.
From that point, each intermediary node in the path first checks the validity of the path via the public keys included and his own identifier, then forwards the transaction including his public key to the gradient node.

Once transactions are received by the round leader, he includes the valid ones into the block.
The block consists of the credential, hash of the previous block and valid transactions with their paths.
Then, the block is propagated throughout the network.

\alglanguage{pseudocode}
\begin{algorithm}[htp!]
\begin{algorithmic}
%\Require Leader credential $\mathcal{L}^r$
%\Ensure For each node, closest node to the leader $cn_i$, for $i=1,\ldots,n$
\State
\State \underline{Recognition Phase}
\State Leader $l$ propagates $\mathcal{L}^r$
\For{Each node $n_i$}
\If{First time receiving $\mathcal{L}^r$ and $IN_{n'}$}
\If{$\mathcal{L}^r$ is valid}
\State Assign $IN_{n_i}\leftarrow IN_{n'}$ and gradient node as $n'$
\State Compute $OUT_{n_i}=H(IN_{n_i},PK_{n_i})$%, where $H_{n_i}=H(R_{n_i},PK_{n_i})$
\State Propagate $\mathcal{L}^r$ and $OUT_{n_i}$ to neighbors.
\EndIf
%\If{Client of $c_T$}
%\State Store the first $N_{out}$ incoming $IN$ values
%\EndIf
\EndIf
\EndFor
\State 
\State \underline{Transaction Phase}
\State  Client $c_T$ provides $Signed(T,IN_{c_T})$ (and $\mathcal{PK}=\emptyset$) to the first $N_{con}$ gradient nodes.
\For{Each node $n_i$ receiving $Signed(T,IN_{c_T})$ and $\mathcal{PK}$}
\If{First time receiving $T$}
\If{Signature path holds}
\State Update $\mathcal{PK}\leftarrow \mathcal{PK} \bigcup \{ PK_{n_i} \}$
\State Send $Signed(T,IN_{c_T})$ and $\mathcal{PK}$ to the gradient node.
\EndIf
\EndIf
\EndFor
\end{algorithmic}
\caption{The Combined Propagation Algorithm} 
\label{alg:combined}
\end{algorithm}

%\begin{figure}
%	\begin{center}
%		\includegraphics[width=\textwidth]{figures/a-pow.pdf}
%		\caption{Illustration of a path between the round leader and a client in A-PoW. Upper-side and solid arrows represent the recognition phase while the lower side and dashed arrows do the transaction phase.}
%	\end{center}
%\end{figure}

\noindent\textbf{Incentive for block propagation.}
As a consequence of the incentive and routing mechanisms, intermediary nodes also have incentives to propagate the block since they share processing fees. 
Even more, the ones who are closer to the leader would have higher motivation since they probably gain from more transactions.
%Therefore, it encourages even more to share rather than selfish mining.

\noindent\textbf{Storage efficiency.}
Any propagation incentive mechanism requires additional data storage than the data itself to keep track of the propagation path.
Previous works having incentive \cite{abraham2016solidus,babaioff2012bitcoin} utilize signature chains where each node signs the transaction and the public key of the receiver.
Therefore, additional to the transaction, the signature package of each propagating node is included.
On the other hand, our solution with the path identification benefits from the recognition phase of the routing protocol, and its additional storage requirement is only the public keys of propagating nodes and a signature of the client.
Since the ability to claim propagation reward and the validation of the path need to be available, our propagation mechanism demands minimal storage components.

\noindent\textbf{Privacy of the intermediary nodes.}
Signature chains and the proposed path identifier yield a direct connection between nodes network ID and their public keys. 
Unlike signature chains, our solution consists of two phases and the propagating nodes validate it by checking whether their input is preserved or not.
This enables us to tackle the privacy issue by replacing plain public keys with commitments.
Instead of directly including a public key, each node can obscure it in a simple commitment with a random number ($CT_i=H(PK_i,R_i)$).
All verifications can be handled with the commitments while claiming propagation reward requires to reveal it.
The commitment version uses the same network structure without compromising the identities of the nodes except clients and the round leader. The location of the round leader and clients will be known to their neighbors.
They may need to update their key pairs or replace their connections for the next rounds.
%The path is stored like a binary tree, and each contributor can demand his fee by proof of his public key and random value matching with his commitment.
%Irrelevant parties may just store the root of the tree.

%\paragraph{\todo{Microblock Case.}}
%We refer Bitcoin-NG for microblocks type implementation where the leader is validated by the PoW in key block, then he propose several microblocks including only transactions.
%In A-PoW, microblocks can be validated by  
%
%\todo{For one block, include PK and signed whole block with SK.
%For multiple case, $H_0 < H^i(H_0)$, for each microblock include HASH\_PREV and $H^2k-i(H_0)$.
%Only the last one can be manipulated.
%In this way, Pk of round leader is never revealed.}
\section{Conclusion}\label{sec:conc}
In this work, we investigated two transaction propagation related problems of block\-chains: incentive and bandwidth efficiency. 
We presented an incentive mechanism encouraging nodes to propagate messages, and a routing mechanism reducing the redundant communication cost. 

We analyzed the necessary and sufficient conditions providing an incentive to propagate messages as well as to deviate participants (nodes) from introducing Sybil nodes.
We studied different types of network topologies and we showed the impossibility result of the Sybil-proofness for the 1-connected model.
We formulated the incentive-compatible propagation mechanism and proved that it obeys the rational behavior.
%We want to stress that the incentive mechanism is independent of the consensus algorithm and works with any consensus protocol.
%Moreover, we expect that the incentive mechanism will be adopted by the blockchain solutions as well as  peer-to-peer systems for peer discovery and advertisement purposes.

We presented a new aspect of the consensus algorithms, namely first-leader-then-block protocols.
We proposed a smart routing mechanism for these protocols, which reduces the redundant transaction propagation from the size of the network to the scale of average shortest path length.
Finally, we combined incentive and routing mechanisms in a compatible and memory-efficient way.

%Overall, in order to have sustainable and purely decentralized blockchain, there should be an incentive for information propagation.

%We hope that the proposed incentive and routing mechanisms help blockchain technology to compete with centralized systems as well as lead to new instruments and applications in blockchain technology.
\noindent\textbf{Future work and open questions.} 
In Section \ref{sec:inc_discussion}, we mentioned the parameter choice and possible outcomes of the incentive mechanism.
Detailed effect of incentive model and parameter choice are left as a future work.
Another open question is the effect of the incentive mechanism on the topology of the network.
Nodes would benefit from increasing their connection to contribute more transaction propagations, i.e., it would increase the connectivity of the network.
Using that result, a rigorous analysis on the choice of the $C$ parameter can be done.
Finally, there are open problems regarding the economics of the transaction fee: analyzing the accuracy of the de facto formulas in the existing cryptocurrencies with respect to the cost of the propagation and validation and investigating the possible impacts of the sharing fee like decentralization effect.
\section{Acknowledgment}
This work was supported by NWO Grant 439.16.614 Blockchain and Logistics Innovation.

\bibliographystyle{IEEEtran}
\bibliography{references}

\end{document}